\documentclass[letterpaper]{article}

\usepackage[utf8]{inputenc}
\usepackage[usenames,dvipsnames]{color}
\usepackage[hyphens]{url}
\usepackage{hyperref}
\hypersetup{
    colorlinks=true,
    linkcolor=OrangeRed,
    citecolor=PineGreen,
    filecolor=Melon,
    urlcolor=DarkOrchid,
}

\usepackage{amssymb,amsmath,amsthm,stmaryrd,mathrsfs,wasysym,mathtools,tikz,tikz-cd} 

\makeatletter
\newcommand{\sectionNotes}{\phantomsection\section*{Notes}\addcontentsline{toc}{section}{Notes}\markright{\textsc{\@chapapp{} \thechapter{} Notes}}}
\newcommand{\sectionExercises}[1]{\phantomsection\section*{Exercises}\addcontentsline{toc}{section}{Exercises}\markright{\textsc{\@chapapp{} \thechapter{} Exercises}}}
\makeatother

\newcommand{\defeq}{\vcentcolon\equiv}  

\makeatletter
\def\prd#1{\@ifnextchar\bgroup{\prd@parens{#1}}{%
    \@ifnextchar\sm{\prd@parens{#1}\@eatsm}{%
    \@ifnextchar\prd{\prd@parens{#1}\@eatprd}{%
    \@ifnextchar\;{\prd@parens{#1}\@eatsemicolonspace}{%
    \@ifnextchar\\{\prd@parens{#1}\@eatlinebreak}{%
    \@ifnextchar\narrowbreak{\prd@parens{#1}\@eatnarrowbreak}{%
      \prd@noparens{#1}}}}}}}}
\def\prd@parens#1{\@ifnextchar\bgroup%
  {\mathchoice{\@dprd{#1}}{\@tprd{#1}}{\@tprd{#1}}{\@tprd{#1}}\prd@parens}%
  {\@ifnextchar\sm%
    {\mathchoice{\@dprd{#1}}{\@tprd{#1}}{\@tprd{#1}}{\@tprd{#1}}\@eatsm}%
    {\mathchoice{\@dprd{#1}}{\@tprd{#1}}{\@tprd{#1}}{\@tprd{#1}}}}}
\def\@eatsm\sm{\sm@parens}
\def\prd@noparens#1{\mathchoice{\@dprd@noparens{#1}}{\@tprd{#1}}{\@tprd{#1}}{\@tprd{#1}}}
\def\lprd#1{\@ifnextchar\bgroup{\@lprd{#1}\lprd}{\@@lprd{#1}}}
\def\@lprd#1{\mathchoice{{\textstyle\prod}}{\prod}{\prod}{\prod}({\textstyle #1})\;}
\def\@@lprd#1{\mathchoice{{\textstyle\prod}}{\prod}{\prod}{\prod}({\textstyle #1}),\ }
\def\tprd#1{\@tprd{#1}\@ifnextchar\bgroup{\tprd}{}}
\def\@tprd#1{\mathchoice{{\textstyle\prod_{(#1)}}}{\prod_{(#1)}}{\prod_{(#1)}}{\prod_{(#1)}}}
\def\dprd#1{\@dprd{#1}\@ifnextchar\bgroup{\dprd}{}}
\def\@dprd#1{\prod_{(#1)}\,}
\def\@dprd@noparens#1{\prod_{#1}\,}

\def\@eatnarrowbreak\narrowbreak{%
  \@ifnextchar\prd{\narrowbreak\@eatprd}{%
    \@ifnextchar\sm{\narrowbreak\@eatsm}{%
      \narrowbreak}}}
\def\@eatlinebreak\\{%
  \@ifnextchar\prd{\\\@eatprd}{%
    \@ifnextchar\sm{\\\@eatsm}{%
      \\}}}
\def\@eatsemicolonspace\;{%
  \@ifnextchar\prd{\;\@eatprd}{%
    \@ifnextchar\sm{\;\@eatsm}{%
      \;}}}

\def\lam#1{{\lambda}\@lamarg#1:\@endlamarg\@ifnextchar\bgroup{.\,\lam}{.\,}}
\def\@lamarg#1:#2\@endlamarg{\if\relax\detokenize{#2}\relax #1\else\@lamvar{\@lameatcolon#2},#1\@endlamvar\fi}
\def\@lamvar#1,#2\@endlamvar{(#2\,{:}\,#1)}
\def\@lameatcolon#1:{#1}

\def\lamu#1{{\lambda}\@lamuarg#1:\@endlamuarg\@ifnextchar\bgroup{.\,\lamu}{.\,}}
\def\@lamuarg#1:#2\@endlamuarg{#1}

\def\fall#1{\forall (#1)\@ifnextchar\bgroup{.\,\fall}{.\,}}

\def\exis#1{\exists (#1)\@ifnextchar\bgroup{.\,\exis}{.\,}}

\def\sm#1{\@ifnextchar\bgroup{\sm@parens{#1}}{%
    \@ifnextchar\prd{\sm@parens{#1}\@eatprd}{%
    \@ifnextchar\sm{\sm@parens{#1}\@eatsm}{%
    \@ifnextchar\;{\sm@parens{#1}\@eatsemicolonspace}{%
    \@ifnextchar\\{\sm@parens{#1}\@eatlinebreak}{%
    \@ifnextchar\narrowbreak{\sm@parens{#1}\@eatnarrowbreak}{%
        \sm@noparens{#1}}}}}}}}
\def\sm@parens#1{\@ifnextchar\bgroup%
  {\mathchoice{\@dsm{#1}}{\@tsm{#1}}{\@tsm{#1}}{\@tsm{#1}}\sm@parens}%
  {\@ifnextchar\prd%
    {\mathchoice{\@dsm{#1}}{\@tsm{#1}}{\@tsm{#1}}{\@tsm{#1}}\@eatprd}%
    {\mathchoice{\@dsm{#1}}{\@tsm{#1}}{\@tsm{#1}}{\@tsm{#1}}}}}
\def\@eatprd\prd{\prd@parens}
\def\sm@noparens#1{\mathchoice{\@dsm@noparens{#1}}{\@tsm{#1}}{\@tsm{#1}}{\@tsm{#1}}}
\def\lsm#1{\@ifnextchar\bgroup{\@lsm{#1}\lsm}{\@@lsm{#1}}}
\def\@lsm#1{\mathchoice{{\textstyle\sum}}{\sum}{\sum}{\sum}({\textstyle #1})\;}
\def\@@lsm#1{\mathchoice{{\textstyle\sum}}{\sum}{\sum}{\sum}({\textstyle #1}),\ }
\def\tsm#1{\@tsm{#1}\@ifnextchar\bgroup{\tsm}{}}
\def\@tsm#1{\mathchoice{{\textstyle\sum_{(#1)}}}{\sum_{(#1)}}{\sum_{(#1)}}{\sum_{(#1)}}}
\def\dsm#1{\@dsm{#1}\@ifnextchar\bgroup{\dsm}{}}
\def\@dsm#1{\sum_{(#1)}\,}
\def\@dsm@noparens#1{\sum_{#1}\,}

\def\wtype#1{\@ifnextchar\bgroup%
  {\mathchoice{\@twtype{#1}}{\@twtype{#1}}{\@twtype{#1}}{\@twtype{#1}}\wtype}%
  {\mathchoice{\@twtype{#1}}{\@twtype{#1}}{\@twtype{#1}}{\@twtype{#1}}}}
\def\lwtype#1{\@ifnextchar\bgroup{\@lwtype{#1}\lwtype}{\@@lwtype{#1}}}
\def\@lwtype#1{\mathchoice{{\textstyle\mathsf{W}}}{\mathsf{W}}{\mathsf{W}}{\mathsf{W}}({\textstyle #1})\;}
\def\@@lwtype#1{\mathchoice{{\textstyle\mathsf{W}}}{\mathsf{W}}{\mathsf{W}}{\mathsf{W}}({\textstyle #1}),\ }
\def\twtype#1{\@twtype{#1}\@ifnextchar\bgroup{\twtype}{}}
\def\@twtype#1{\mathchoice{{\textstyle\mathsf{W}_{(#1)}}}{\mathsf{W}_{(#1)}}{\mathsf{W}_{(#1)}}{\mathsf{W}_{(#1)}}}
\def\dwtype#1{\@dwtype{#1}\@ifnextchar\bgroup{\dwtype}{}}
\def\@dwtype#1{\mathsf{W}_{(#1)}\,}

\def\wtypeh#1{\@ifnextchar\bgroup%
  {\mathchoice{\@lwtypeh{#1}}{\@twtypeh{#1}}{\@twtypeh{#1}}{\@twtypeh{#1}}\wtypeh}%
  {\mathchoice{\@@lwtypeh{#1}}{\@twtypeh{#1}}{\@twtypeh{#1}}{\@twtypeh{#1}}}}
\def\lwtypeh#1{\@ifnextchar\bgroup{\@lwtypeh{#1}\lwtypeh}{\@@lwtypeh{#1}}}
\def\@lwtypeh#1{\mathchoice{{\textstyle\mathsf{W}^h}}{\mathsf{W}^h}{\mathsf{W}^h}{\mathsf{W}^h}({\textstyle #1})\;}
\def\@@lwtypeh#1{\mathchoice{{\textstyle\mathsf{W}^h}}{\mathsf{W}^h}{\mathsf{W}^h}{\mathsf{W}^h}({\textstyle #1}),\ }
\def\twtypeh#1{\@twtypeh{#1}\@ifnextchar\bgroup{\twtypeh}{}}
\def\@twtypeh#1{\mathchoice{{\textstyle\mathsf{W}^h_{(#1)}}}{\mathsf{W}^h_{(#1)}}{\mathsf{W}^h_{(#1)}}{\mathsf{W}^h_{(#1)}}}
\def\dwtypeh#1{\@dwtypeh{#1}\@ifnextchar\bgroup{\dwtypeh}{}}
\def\@dwtypeh#1{\mathsf{W}^h_{(#1)}\,}

\makeatother

\newcommand{\proj}[1]{\ensuremath{\mathsf{pr}_{#1}}\xspace}
\newcommand{\fst}{\ensuremath{\proj1}\xspace}
\newcommand{\snd}{\ensuremath{\proj2}\xspace}

\newcommand{\rec}[1]{\mathsf{rec}_{#1}}
\newcommand{\ind}[1]{\mathsf{ind}_{#1}}

\newcommand{\refl}[1]{\ensuremath{\mathsf{refl}_{#1}}\xspace}

\newcommand{\UU}{\ensuremath{\mathcal{U}}\xspace}

\newcommand{\typele}[1]{\ensuremath{\mathsf{Type}_{#1}}}

\newcommand{\prop}{\ensuremath{\mathsf{Prop}}\xspace}

\newcommand{\isprop}{\ensuremath{\mathsf{isProp}}}

\newcommand{\trunc}[2]{\mathopen{}\left\Vert #2\right\Vert_{#1}\mathclose{}}

\newcommand{\emptyt}{\ensuremath{\mathbf{0}}\xspace}

\newcommand{\unit}{\ensuremath{\mathbf{1}}\xspace}

\newcommand{\inlsym}{{\mathsf{inl}}}
\newcommand{\inrsym}{{\mathsf{inr}}}
\newcommand{\inl}{\ensuremath\inlsym\xspace}
\newcommand{\inr}{\ensuremath\inrsym\xspace}

\newcommand{\uset}{\ensuremath{\mathcal{S}et}\xspace}

\newbox\pbbox

\newcommand{\y}{\ensuremath{\mathbf{y}}\xspace}

\newcommand{\N}{\ensuremath{\mathbb{N}}\xspace}

\newcommand{\suc}{\mathsf{succ}}

\newcommand{\Z}{\ensuremath{\mathbb{Z}}\xspace}

\newcommand{\funext}{\mathsf{funext}}

\newcommand{\oftp}[3]{#1 \vdash #2 : #3}

\newcommand{\tmtp}[2]{#1 \mathord{:} #2}

\newcommand{\rintro}{\textsc{intro}}
\newcommand{\relim}{\textsc{elim}}

\newcommand{\rid}{\textsc{id}}

\newcommand{\isType}[1]{\mathsf{isType}_{#1}}

\makeatletter

\newtheoremstyle{exampletheoremstyle}
{3pt}
{3pt}
{}
{}
{\itshape}
{:}
{.5em}
{}

\newtheoremstyle{mydef}
{\topsep}
{\topsep}%
{\setlist{itemsep=0pt}}
{}
{\bfseries}
{. \vspace{0.9ex}}
{\newline}
{%
\thmname{#1}~\thmnumber{#2}\thmnote{\ -\ #3}%
}%

\theoremstyle{mydef} 
  {\pushQED{\qed}\mydefn}
  {\popQED\endmydefn}
\theoremstyle{remark}

\let\c@equation\c@thm
\numberwithin{equation}{section}

\def\noteson{%
\gdef\note##1{\mbox{}\marginpar{\color{blue}\textasteriskcentered\ ##1}}}

\noteson

\newcounter{symindex}

\newcommand{\catj}{\mathbb{J}}
\newcommand{\catset}{\uset}

\newcommand{\falg}{F\text{-}\mathbf{Alg}}
\newcommand{\talg}{T\text{-}\mathbf{Alg}}
\makeatletter
\newcommand*{\relrelbarsep}{.386ex}
\newcommand*{\relrelbar}{%
	\mathrel{%
		\mathpalette\@relrelbar\relrelbarsep
	}%
}
\newcommand*{\@relrelbar}[2]{%
	\raise#2\hbox to 0pt{$\m@th#1\relbar$\hss}%
	\lower#2\hbox{$\m@th#1\relbar$}%
}
\providecommand*{\rightrightarrowsfill@}{%
	\arrowfill@\relrelbar\relrelbar\rightrightarrows
}
\providecommand*{\leftleftarrowsfill@}{%
	\arrowfill@\leftleftarrows\relrelbar\relrelbar
}
\providecommand*{\xrightrightarrows}[2][]{%
	\ext@arrow 0359\rightrightarrowsfill@{#1}{#2}%
}
\providecommand*{\xleftleftarrows}[2][]{%
	\ext@arrow 3095\leftleftarrowsfill@{#1}{#2}%
}
\makeatother

\newcommand{\qtext}[1]{\quad\text{#1}\quad}
\newcommand{\ent}{\vdash}
\newcommand{\jcomment}[1]{}
\newcommand{\scomment}[1]{}
\newcommand{\slscomment}[1]{}
\newcommand{\recpf}{\rec{+}^F}
\newcommand{\indpf}{\ind{+}^F}
\newcommand{\recps}{\rec{+}}
\newcommand{\recp}{\rec{+}}
\newcommand{\indps}{\ind{+}}
\newcommand{\inlf}{\inl^F}
\newcommand{\inrf}{\inr^F}
\newcommand{\inls}{\inl}
\newcommand{\inrs}{\inr}
\newcommand{\recn}{\rec{\N}}

\newcommand{\itprd}[1]{\textstyle\prod_{\{#1\}}}

\newcommand{\emb}{\hookrightarrow}
\newcommand{\natf}{\mathbb{N}^*}
\newcommand{\zerof}{0^*}
\newcommand{\sucf}{\suc^*}
\newcommand{\recnf}{\rec{\mathbb{N}}^*}

\newcommand{\sone}{\mathsf{S}^1}
\newcommand{\rece}{\rec{\emptyt}}
\newcommand{\vxymatrix}[1]{\vcenter{\xymatrix{#1}}}

\usepackage{booktabs}   
\usepackage{subcaption} 
\usepackage{cleveref}
\usepackage{mdframed}
\usepackage{mathpartir}
\usepackage{xspace}

\usepackage[all,cmtip]{xy}
\xyoption{2cell}
\xyoption{curve}
\UseTwocells

\newcommand{\Prop}{\ensuremath{\mathsf{Prop}}} 
\newcommand{\isProp}{\ensuremath{\mathsf{isProp}}} 
\newcommand{\Set}{\ensuremath{\mathsf{Set}}} 
\newcommand{\isSet}{\ensuremath{\mathsf{isSet}}} 
\newcommand{\set}{\ensuremath{\mathsf{Set}}} 

\newcommand{\Id}{\mathtt{Id}}

\newcommand{\truncation}[1]{|\!| #1 |\!|}

\theoremstyle{plain}
\newtheorem{theorem}{Theorem}[section]

\theoremstyle{definition}

\theoremstyle{remark}
\newtheorem*{remark}{Remark}

\theoremstyle{plain}
\newtheorem{corollary}{Corollary}[section]

\title{\textbf{Impredicative Encodings of (Higher) Inductive Types}}
\date{\today}
\author{
  Steve Awodey\\
  \texttt{awodey@cmu.edu}
  \and
  Jonas Frey\\
  \texttt{jonasf@andrew.cmu.edu}
  \and
  Sam Speight\\
  \texttt{samuel.speight@cs.ox.ac.uk}
}

\begin{document}

\maketitle

\begin{abstract}
Postulating an impredicative universe in dependent type theory allows System F style encodings of finitary inductive types, but these fail to satisfy the relevant $\eta$-equalities and consequently do not admit dependent eliminators. To recover $\eta$ and dependent elimination, we present a method to construct \emph{refinements} of these impredicative encodings, using ideas from homotopy type theory.
We then extend our method to construct impredicative encodings of some \emph{higher} inductive types, such as $1$-truncation and the unit circle $\sone$.
\end{abstract}


\section{Introduction}

System F, also known as the `Girard-Reynolds polymorphic $\lambda$-calculus', goes back to \cite{girardoriginal} and \cite{reynoldsoriginal}. It extends the simply typed $\lambda$-calculus with universal quantification $\forall$ over types. Under the Curry-Howard correspondence \cite{howard:pat}, 
it is the type-theoretic analog of second-order propositional logic.


One of the remarkable things about System F is that it allows for the \emph{encoding} of types such as products, sums, natural numbers, and, more generally, \emph{finitary inductive types}. These encodings are called \emph{impredicative}, since in defining a specific type they quantify over the totality of \emph{all} types, which in particular contains the type which is  being defined.
For example, the type $\mathbb{N}$ of natural numbers is encoded in System F as
\[
\N_F\ \equiv\ \forall \,X.\, (X\rightarrow X)\rightarrow X\rightarrow X\,.
\]	
A well-known defect of such encodings, however, is that they do not satisfy the appropriate `$\eta$-rules', which are uniqueness principles stating that for every recursive definition there is only one function realizing it. 
One way to think about this failure of $\eta$ is that e.g.\ the type described by the formula $\N_F$ is `too large',  in that certain models~\cite{rummelhoff2004polynat} may contain non-standard elements which are not generated from the constructors. On the other hand, one can show using \emph{parametricity} arguments~\cite{reynolds1983types} that all the `named' elements are standard. This discrepancy has led to attempts to \emph{refine} the models by imposing parametricity, \emph{dinaturality}~\cite{BAINBRIDGE199035}, or \emph{realizability} \cite{carboni1988categorical} conditions.

System F-style, impredicative encodings can also be given in dependent type theory with an \emph{impredicative universe} (such as the \emph{calculus of constructions}~\cite{COQUAND198895}), but here a further consequence of the failure of uniqueness of the recursors is that the encoded types do not admit {dependent elimination rules}, which are necessary for proofs by induction, and thus indispensable for the development of mathematics in type theory.

In this article we present a new technique to restore $\eta$-rules by refining impredicative encodings in a way that is related to the parametricity and dinaturality techniques mentioned above, but in contrast to them, takes place \emph{inside the type theory}, so that the $\eta$-equalities on the refined types become \emph{provable}, rather than just admissible w.r.t.\ a model. To be more specific, for our refined encodings we can prove \emph{propositional} versions of the required $\eta$-equalities, which -- as shown in~\cite{ags:it-hott,AGS2} -- are sufficient to derive the existence of the corresponding dependent eliminators.

The system of type theory that we use is similar to the calculus of constructions with a hierarchy of predicative universes (like in older versions of the Coq proof assistant), but
in contrast to the calculus of constructions we assume that the lowest, impredicative universe is closed under small sums and identity types.

Some of our arguments make use of  the `uniqueness of identity proofs' principle, but instead of postulating it globally we exploit the notion of \emph{$0$-types} from homotopy type theory to state the relevant results for that level of the definable hierarchy of h-levels.
This also allows us to explore the applicability of our techniques to \emph{higher inductive types} in the later sections.

We emphasize that our main contribution is to give impredicative encodings of  inductive types in dependent type theory that satisfy the relevant dependent elimination rules (along with the other rules).
This seems to solve a long-standing problem, considered in \cite{10.1007/3-540-45413-6_16}, of giving ``2nd-order encodings''
 satisfying induction.  We do this by  ``refining'' the impredicative encodings inside the type theory using identity types. 
This is in contrast to interpreting the system into a  model and having the universal properties hold there, which is the spirit of the ``parametric polymorphism'' approach set out in \cite{reynolds1983types}, where the universal quantification of System F is ``cut down'' on interpretation.
\paragraph*{Overview}
Section \ref{sec_system} specifies the system of type theory in which our work takes place, recalls the definition of $n$-types in homotopy type theory, and introduces the (pre-)category $\catset$ of small $0$-types. We also include a brief description of  how System F can be translated into our setting.
In Section~\ref{sec:basic} we show how to refine the System F encodings of some non-recursive, inductive sets such as the binary sum $A+B$ of two sets to recover appropriate $\eta$-rules, using an argument based on a type-theoretic version of the Yoneda lemma.
Section~\ref{sec:generalz} gives a related technique
involving impredicative encodings of initial algebras  in order to achieve the same result for genuinely inductive sets such as the type $\N$ of natural numbers. 
In Section~\ref{sec:1types} we show how these techniques generalize from sets ($0$-types) to $1$-types, giving impredicative encodings of some of the recently introduced \emph{higher inductive types} \cite[Ch.~6]{HoTTbook1}. Specifically, we give encodings for the $1$-truncation and the unit circle. We believe that these encodings also illuminate the ones for conventional inductive types. Indeed, our general methodology is very much informed by the HoTT point of view.
Finally, Section~\ref{sec:limits} addresses issues such as limitations of our techniques, open questions, consistency and the existence of models, and future work.

%

\section{The System of Type Theory}\label{sec_system}

Although our results are not all ``higher dimensional'' in nature, our work is best understood in the context of \emph{homotopy type theory (HoTT)}; we refer to~\cite{HoTTbook1} as our standard reference for conventions and terminology. We thus work in a system of dependent type theory with products $\prod_{x:A} B(x)$, strong sums $\sum_{x:A} B(x)$, (intensional) identity types $\Id_X(x,y)$, and function extensionality, as in~\cite{HoTTbook1}.  However, we  make no use of the univalence axiom. We usually write simply $x=y$ for $\Id_X(x,y)$, as is now common. We then distinguish notationally between \emph{propositional} equality $x = y$ and \emph{definitional} equality $x\equiv y$.

\paragraph*{Universes} We augment the hierarchy of \emph{predicative} universes 
\[\UU_0: \UU_1: \UU_2: \dots\] 
assumed in \cite{HoTTbook1} by adding a {single \emph{impredicative} universe}  $\UU$ at the bottom.
\[\UU:\UU_0: \UU_1: \UU_2: \dots\]
This new universe $\UU$ is also closed under dependent sums and identity types, like the $\UU_i$, but instead of the usual (predicative) product formation rule
\[
	\inferrule*
{{\Gamma \vdash A:\UU_i } \\ \oftp{\Gamma, x:A}{B}{\UU_i } }
{\oftp\Gamma{\tprd{x:A}B}{\UU_i }}
\]
it satisfies the \emph{impredicative} product formation rule
\[
\inferrule*
{\oftp{\Gamma, x:A}{B}{\UU} }
{\oftp\Gamma{\tprd{x:A}B}{\UU}}
\]
which is stronger since there is no size restriction on the $A$.

Note that in the following, when writing $\UU$, we always mean the \emph{impredicative} universe -- in this respect we deviate from \cite{HoTTbook1}, in that we do \emph{not} use $\UU$ as a placeholder for an unspecified $\UU_i$.

\paragraph*{n-Types} Recall from \cite[7.1]{HoTTbook1} the hierarchy of n-types (Voevodsky: ``h-levels''): $X$ is called a \emph{(-1)-type}, or \emph{proposition}, if it satisfies $\tprd{x,y:X}x=y$; it is a \emph{0-type}, or \emph{set}, if its identity types are always propositions; and generally, it is an \emph{(n+1)-type} if its identity types are always n-types. Formally, let
\begin{align*}
\isProp(X)\ :\equiv&\ \tprd{x,y:X}\, x=y\\
\isSet(X)\ :\equiv&\ \tprd{x,y:X}\, \isProp(x=y)\\
\isType{(n+1)}(X)\ :\equiv&\ \tprd{x,y:X}\, \isType{n}(x=y)\,,
\end{align*}
and
\begin{align*}
\Prop\ :\equiv&\ \tsm{X:\UU}\ \isProp(X) \\
\Set\ :\equiv&\ \tsm{X:\UU}\ \isSet(X) \\
\typele{n}\ :\equiv&\ \tsm{X:\UU}\ \isType{n}(X)\,.
\end{align*}

Note that the types $\isProp(X)$, $\isSet(X)$, etc., are themselves propositions~\cite[Theorem~7.1.10]{HoTTbook1}, so that $\Prop$, $\Set$, etc., are \emph{subtypes} of $\UU$, in the sense that the first projection from the respective $\Sigma$-type is an embedding. 

We normally suppress the coercion $\fst:\typele{n}\to\UU$ and treat expressions of type $\typele{n}$ as if they were themselves types. Thus, in particular if $X:\UU\ent A(X)$ is a $\UU$-indexed family of types then the expression $\tprd{X:\typele{n}}A(X)$ is a shorthand for 
\[\tprd{X:\typele{n}}\,A(\fst\,X),\] 
which in turn is equivalent to 
\[\tprd{X:\UU}\,\isType{n}(X)\to A(X).\]

Moreover, since the $n$-types are closed under arbitrary products~\cite[Theorem~7.1.9]{HoTTbook1} (and suppressing the unpacking and repacking of dependent pairs), the rule
\[
\inferrule*
{\oftp{\Gamma,x:A}{B}{\typele{n}} }
{\oftp\Gamma{\tprd{x:A}B}{\typele{n}}}
\]
is admissible for all $n$. Thus, in sum, we can view the types $\typele{n}$ as \emph{impredicative 
subuniverses} of $\UU$.

We exploit the convenience of having an impredicative universe that is closed under most constructions, by working mostly inside $\UU$ -- as opposed to the usual methodology of predicative systems, where a hierarchy of universes are used ``parametrically''. Accordingly, we adopt the convention that the terms ``proposition'', ``set'', ``n-type'', etc., refer only to types in $\UU$. 

\paragraph*{The precategory $\catset$ of sets in $\UU$}

The subtype $\set$ of $0$-types in $\UU$ gives rise to a  {precategory} $\uset$ where 
\[\uset_0=\set \quad\text{and}\quad \hom(A,B)=(A\to B)\] 
for $A,B: \set$ (\cite[Example~9.1.5]{HoTTbook1}). 

As pointed out in~\cite[Section~10.1.1]{HoTTbook1}, this precategory is (small) \emph{complete} in that it admits equalizers (defined using $\Sigma$- and identity-types) and small products (given by type-theoretic products).  Since in our setting $\UU$ is an impredicative universe, $\catset$ even admits `large' products -- i.e.\ products indexed by arbitrary types -- which we make use of in what follows.

\paragraph*{Translation of System F}
There is an evident syntactic translation $t$ from System F (see Appendix \ref{appendix_F}) to our system of dependent type theory. Types of the form $A \rightarrow B$ in System F are translated to $A^t \rightarrow B^t$, and types of the form $\forall X. B$ are translated to $\prod_{X:\UU}B^t$, where $A^t$ and $B^t$ are the translations of the System F types $A$ and $B$. The translation of terms is equally obvious.

Similarly, we can restrict the translation by replacing $\UU$ above with any of the impredicative subuniverses $\typele{n}$ of $n$-types.  For example, we may define the translation $t_0$ with:
\[
(\forall X. B)^{t_0} \defeq \prod_{X:\Set}B^{t_0}.
\]
In this sense, we can speak of System F encodings of inductive types in our system of dependent type theory.

Generally, if $x:A \vdash P(x)$ is a (family of) propositions, then as above $\sum_{x:A}P(x)$ is a subtype of $A$ via the first projection.  Our impredicative encodings of inductive types will be subtypes of the usual System F encodings in this sense; we say that they ``sharpen'' or ``refine'' the usual encodings.

\section{Basic Set Encodings}\label{sec:basic}


As stated in the foregoing section, the impredicative universe $\UU$ allows us to give `System F style' encodings of certain inductive types. In this section we explain how these encodings fall short of the usual inductive types assumed in dependent type theory. We then indicate a way to remedy these shortcomings to a certain extent. 

We start with the sum of two types $A$ and $B$, whose System F encoding we translate into type theory, as explained above, by replacing the quantification over types by a dependent product over $\UU$:
\begin{equation}\label{eq:Fplus}
A +^F\!B\ :\equiv\ \prod_{X:\,\UU} (A\rightarrow X)\rightarrow(B\rightarrow X)\rightarrow X\,.
\end{equation}
It is easy to see that this encoding admits injections\footnote{We adopt the convention that arguments in braces $\{-\}$ in types denote \emph{implicit} arguments, meaning that we may write e.g.\ $\inlf a$ instead of $\inlf A\,B\,a$.}
\begin{equation}\label{eq:finj}
\begin{aligned}
\inlf\ &:\equiv\ \lam{\,A\,B\,a\, X\, f\, g}f\,a\ : \ \textstyle\prod_{\{A, B:\UU\}}A\to (A+^F\!B)\\
\inrf\ &:\equiv\ \lam{\,A\,B\,b\, X\, f\, g}g\,b\ : \ \textstyle\prod_{\{A, B:\UU\}}B\to (A+^F\!B)
\end{aligned}
\end{equation}
and a recursor
\begin{equation}\label{eq:frec}
\begin{aligned}
\recpf\ &:\equiv\ \lam{\,A\,B\,C\,f\, g\, \phi}\phi\,C\,f\,g \\
&: \itprd{A,B,C:\UU}(A\to C)\to (B\to C)\to (A+^F\!B)\to C
\end{aligned}
\end{equation}
satisfying the \emph{definitional} $\beta$-equalities 
\[\recpf\,f\,g\,(\inlf\, a)\equiv f\,a\qtext{and}\recpf\,f\,g\,(\inrf\,b)\equiv g\,b\] 
for all $A,B,C:\UU$, $f:A\to C$, $g:B\to C$, $a:A$ and $b:B$. 

However there are several problems:
\begin{enumerate}
	\item\label{one} the recursor only allows us to define functions into types in~$\UU$,
	\item\label{two} the $\eta$-rule
	\begin{equation*}\label{eqn:eta+}
	f = \recpf(f\circ \inlf)(f\circ\inrf),
	\end{equation*}
	where $f:A+^F\!B\to C$, doesn't  hold, even \emph{propositionally},
	\item\label{three} the encoding does not admit a dependent eliminator, which would have to have the type
	\begin{align*}
\indpf\ :\ 
& \textstyle\prod\{A,B:\UU\}\\
& \textstyle\prod\{C:A+^F\!B\to\UU\}\\
& \textstyle\prod(f:\tprd{a:A}C(\inlf\, a))\\
& \textstyle\prod(g:\tprd{b:B}C(\inrf\, b))\\
& \textstyle\prod(x:A+^F\!B),\quad C\,x.
\end{align*}
	and satisfy the propositional equalities
\[\quad\indpf\,f\,g\,(\inlf a)\ =\ f\,a\qtext{and}\indpf\,f\,g\,(\inrf b)\ =\ g\,b\] 
for all appropriately typed $f$, $g$, $a$, and $b$.
%
\end{enumerate}

We defer discussion of issue \ref{one} in general to section \ref{sec:limits} below. 
Issues \ref{two} and \ref{three} are related by the general theory developed in \cite{ags:it-hott,AGS2}: briefly, in the present setting the dependent elimination rule is equivalent to the $\eta$-rule. In the following we give a way to restore the propositional $\eta$-rule for sums of sets by restricting the product in~\eqref{eq:Fplus} to $\set$ and taking a suitable subtype.

\subsection{Refining the encoding}

Restricting the dependent product in~\eqref{eq:Fplus} to the subuniverse of propositions, we obtain a well-known encoding of logical disjunction:
\[
A \vee B\ \simeq\ \prod_{X:\prop} (A\rightarrow X)\rightarrow(B\rightarrow X)\rightarrow X \qquad\text{for } A,B:\prop
\]
Observe that $A\vee B$ \emph{is} a proposition, because $\prop$ is closed under $\Pi$-types.

In a similar vein it seems natural to define a sum operation
\begin{equation}\label{eq:plus*}
A +^* B\ :\equiv\ \prod_{X:\,\set} (A\rightarrow X)\rightarrow(B\rightarrow X)\rightarrow X\quad\text{for } A,B\in\set
\end{equation}
of sets $A$, $B$ by restricting the impredicative product to types $X:\set$. This type also admits injections and a recursor
\begin{equation*}\label{eq:*inj}
\begin{aligned}
\inl^*\ &:\equiv\ \lam{\,A\,B\,a\, X\, f\, g}f\,a\ : \ \textstyle\prod_{\{A, B:\set\}}A\to (A+^*B)\\
\inr^*\ &:\equiv\ \lam{\,A\,B\,b\, X\, f\, g}g\,b\ : \ \textstyle\prod_{\{A, B:\set\}}B\to (A+^*B)\\
\recp^*\ &:\equiv\ \lam{\,A\,B\,C\,f\, g\, \phi}\phi\,C\,f\,g \\
&: \itprd{A,B,C:\set}(A\to C)\to (B\to C)\to (A+^*B)\to C
\end{aligned}
\end{equation*}
analogous to those for $A+^FB$. The $\eta$-rule still fails for $A+^*B$, but as we shall see, it is now possible to carve out a \emph{subtype} for which it is satisfied.

As a warm-up exercise, consider first the unary case. For $A:\set$ there is an embedding-retraction pair
\begin{equation}\label{eq:embedretract}
\vcenter{\xymatrix{
 A \ar[rr]^-{e} \ar[rrd]_{=} && \ar[d]^{r} \tprd{X:\set} (A\rightarrow  X)\rightarrow X\\
 && A\,,
 }}
 \end{equation}
 where $e(a)\ \equiv \lambda X f. f(a)$ and 
$r(\alpha)\ \equiv\ \alpha_A(\mathrm{id}_A)$.
 
Now, a term $\alpha: \prod_{X:\set} (A\rightarrow X)\rightarrow X$ is a \emph{family of maps} (switching notation for emphasis), $$\alpha_X : X^A \to X\,, \qquad X:\set\,.$$  We can cut down the type $\prod_{X:\set} (A\rightarrow  X)\rightarrow X$ to (one equivalent to) the image of $e$ in \eqref{eq:embedretract} by requiring that the family of maps $\alpha_X$ be \emph{natural in $X$} in the sense that for all sets $X, Y$ and all maps $f : X \to Y$, the following square commutes.
\begin{equation}\label{eq:naturality}
\vxymatrix{
X^A \ar[r]^{\alpha_X} \ar[d]_{f^A} & \ar[d]^{f} X \\
Y^A \ar[r]_{\alpha_Y}   & Y
}
\end{equation}
Here, $f^A \equiv \lam{g}f\circ g$ is the action of the  functor $(-)^A : \uset\to\uset$ on $f : X\to Y$.

The \emph{sharper encoding} $A^+$ of $A$ is now:
\begin{equation}\label{eq:encodeA}
A^+\defeq\sum_{\alpha : A^*} N(\alpha)\,,
\end{equation}
\begin{align*}
&\text{where}& A^* &\defeq \prod_{X:\set} (A\rightarrow X)\rightarrow X\\
&\text{and}&N(\alpha)&\defeq\prod_{X,Y:\set} \prod_{f:X\rightarrow Y} \alpha_Y \circ f^A = f\circ \alpha_X\,.
\end{align*}
Note that $\fst:A^+\emb A^*$ is an embedding since $N(\alpha)$ is a proposition for all $\alpha$.

\begin{theorem}[Basic Lemma]\label{lem:basic}
For any set $A$, we have $A\simeq A^+$.
\end{theorem}
\begin{proof}
First, we show that $e : A \emb A^*$ factors through 
$A^+\emb A^*$.
For $a:A$, the family $e(a) : A^*$ consists of the evaluations 
$e(a)_X : A^X \to X$, where $g \mapsto g(a)$. If $f:X\to Y$, then indeed
\[\begin{split}
(e(a)_Y \circ f^A)(g) =  e(a)_Y (f^A)(g)) =  e(a)_Y (f\circ g)\\
= f(g(a)) = f(e(a)_X(g)) = (f\circ \alpha_X)(g) \,.
\end{split}
\]
Now let $\alpha : A^*$ be natural in the sense expressed in \eqref{eq:encodeA}, and define $a_0 :\equiv r(\alpha) \equiv \alpha_A(1_A) : A\,.$  We claim that $\alpha = e(a_0)$, which suffices since $A^*$ is a set. Indeed, take any $X$ and $g:X\to A$, then we have 
\[
e(a_0)_X(g) = g(a_0) =  g(\alpha_A(1_A))
= \alpha_X(g^A(1_A)) = \alpha_X(g)\,,
\]
using the naturality of $\alpha$ in the third step.
\end{proof}

\begin{remark} The categorically minded reader will recognize that the previous theorem is an instance of the Yoneda lemma. Indeed, \eqref{eq:encodeA} is the type of all natural transformations from the (covariant) representable functor $(-)^A$ to the identity functor $I : \uset\to\uset$, so by Yoneda we indeed have 
\[
\mathrm{Nat}\big(\,(-)^A,\, I\,\big) \simeq I(A) \simeq A
\,.
\]
But since we do not require this level of generality here, we will not develop the required details.
\hfill$\diamondsuit$ 
\end{remark}
Taking inspiration from the previous theorem, we return to the preliminary encoding $A+^*B$ from \eqref{eq:plus*}, and in order to recover the $\eta$-rule define a subtype $A+B\hookrightarrow A+^*B$ by imposing a suitable naturality condition. 

We start with the observation that if we already had such a type $A+B$, then $(A+B)^*$ would be equivalent to $(A+^* B)$:
\begin{equation*}\label{eqn:A+B*}
\begin{aligned}
(A+B)^*\ \equiv&\ \prod_{X:\set} ((A+B)\rightarrow X)\rightarrow X\\
\simeq&\ \prod_{X:\set} (A\rightarrow X)\times(B\rightarrow X)\rightarrow X\\
\simeq&\ \prod_{X:\set} (A\rightarrow X)\to(B\rightarrow X)\rightarrow X\ \equiv\ A+^*B
\end{aligned}
\end{equation*}
Now by transporting the naturality condition of the lemma along the equivalence, we can \emph{define} 
$A+B$ as a subtype of $A+^* B$, where the defining condition can again be read as a naturality property, but now one that does not assume the existence of $A+B$. Specifically, we define
\begin{equation}\label{eqn:encodingA+B}
\begin{split}
A+B &\defeq\sum_{\alpha : A+^* B}\,N(\alpha)\qquad\text{where}
\\
N(\alpha)&\defeq\prod_{X,Y:\mathsf{Set}}\,\prod_{f:X\rightarrow Y}\,\prod_{\substack{h:A\to X\\ k:B\to X}}
f(\alpha_Xhk)=\alpha_Y(f\!\circ\! h)(f\!\circ\! k).
\end{split}
\end{equation}
If we substitute $(\inl^*\, a)$ or $(\inr^*\, b)$ for $\alpha$ in $N(\alpha)$, the two sides of the identity type become definitionally equal, whence we can {refine} the injections defined in~\eqref{eq:finj} to get the following.
\begin{equation*} 
\begin{array}{ll}
\inls:\ \textstyle\prod_{\{A, B:\set\}}A\to A+ B
	&\inrs\,:\,\textstyle\prod_{\{A, B:\set\}}B\to A+ B\\
\inls\,a \defeq (\inl^*\,a,\,\lam{XYfhk}\mathsf{refl})&
	\inrs\,b \defeq (\inr^*\,b,\,\lam{XYfhk}\mathsf{refl})
\end{array}
\end{equation*}
The recursor~\eqref{eq:frec} gets replaced by
\begin{equation*}
\begin{aligned}
&\recps\ :\itprd{A,B,C:\set}(A\to C)\to (B\to C)\to A+B\to C\\
&\recps \defeq \lam{\,A\,B\,C\,f\,g\,\xi}\;(\fst\,\xi)_C\,f\,g
\end{aligned}
\end{equation*}
With these definitions we can now prove the following.
\begin{theorem} For all sets $A$ and $B$, the encoding \eqref{eqn:encodingA+B} of the sum $A+B$, along with the structure $\inls$ , $\inrs$, and $\recps$ just defined, satisfy
\begin{enumerate}
\item 
the \emph{definitional} $\beta$-rules
\[\recps\,f\,g\,(\inls\, a)\equiv f\,a\qtext{and}\recps\,f\,g\,(\inrs\,b)\equiv g\,b\] 
for all $C:\set$, $f:A\to C$, $g:B\to C$, $a:A$, $b:B$\,,
\item 
the 
\emph{propositional} $\eta$-rule
\[\recps(f\circ \inls)(f\circ\inrs)=f\]
for all $C:\set$ and $f:A+B\to C$.
\end{enumerate}
\end{theorem}
\begin{proof}
The $\beta$ rules follow mechanically by unfolding definitions. 

For the $\eta$ rule, we first prove a special case, namely 
\begin{equation}\label{eq:weaketa}
\recps\,\inls\,\inrs = \mathrm{id}_{A+B}: A+B \to A+B\,.
\end{equation}
By function extensionality and $\Sigma$-induction  it is sufficient to show that $\alpha_{A+B}\,\inls\,\inrs=(\alpha,p)$ for all $\alpha: A+^* B$ and $p:N(\alpha)$. Since
$A+B\emb A+^*B$ is an embedding, this reduces to $\fst(\alpha_{A+B}\,\inls\,\inrs)=\alpha$, and again by 
function extensionality this follows from $$\fst(\alpha_{A+B}\,\inls\,\inrs)_X\,f\,g=\alpha_X\,f\,g$$
for
$X:\set$, $f: A\to X$ and $g: B\to X$. This is shown by
\begin{align*}
\fst(\alpha_{A+B}\,\inls\,\inrs)_X\,f\,g 
&\equiv \recps fg(\alpha_{A+B}\,\inls\,\inrs)\\
&=\alpha_X(\recps fg\circ\inls)(\recp fg\circ\inrs) \\
&=\alpha_Xfg\,,
\end{align*}
where the second equality is given by $p(\recps fg)\,\inls\,\inrs$ and the third one follows from $\beta$ and function extensionality.

For the general case, let again $\alpha:A+^*B$ and $p:N(\alpha)$. We have:
\begin{align*}
	\recps (f\circ \inls)(f\circ\inrs)(\alpha,p) &\equiv \alpha_C(f\circ \inls)(f\circ\inrs)\\
	&=f(\alpha_C(\inls)(\inrs)) &&(\text{by }p)\\
	&\equiv f(\recps \,\inls\,\inrs\,(\alpha,p))\\
	&=f(\alpha,p)&&(\text{by }\eqref{eq:weaketa})\,,
	\end{align*}
	which proves the claim.
\end{proof}
We emphasize that it is crucial to the proof that the encoding \eqref{eqn:encodingA+B} of $A+B$ is itself a set, so that $A+B$ is in the range of the variable $X:\set$ -- and of course, so that the sum of two sets is again a 0-type. This is ensured by the fact that the n-types are closed under $\Sigma$-, $\Pi$-, and identity-types, and of course, the impredicativity of \UU.  In more detail, in \eqref{eq:plus*}, the $X$ ranges over sets, and thus the type $(A\rightarrow X)\to(B\rightarrow X)\rightarrow X$ is a 0-types.  But then by impredicativity of $\Set$, the entire type $A+^*B$ is a set. In $N(\alpha)$, the identity type is a proposition, since it is over the type of functions $(A\to X)\to(B\to X)\to Y$, which is a 0-type.  Thus $N(\alpha)$ is itself a proposition, whence $A+B$ is a set. 

Finally, as mentioned in \Cref{sec:basic}, the \emph{induction principle} below follows from recursion together with the uniqueness of the recursor (the $\eta$-rule)~\cite{ags:it-hott,AGS2}:
\begin{align*}
\indps\ :\ 
& \textstyle\prod\{A,B:\Set\}\\
& \textstyle\prod\{C:A+B\to\Set\}\\
& \textstyle\prod(f:\tprd{a:A}C(\inls\, a))\\
& \textstyle\prod(g:\tprd{b:B}C(\inrs\, b))\\
& \textstyle\prod(x:A+B),\quad C\,x,
\end{align*}
with propositional equalities
\[\quad\indps\,f\,g\,(\inls\,a)\ =\ f\,a\qtext{and}\indps\,f\,g\,(\inrs\,b)\ =\ g\,b\] 
for all appropriately typed $f$, $g$, $a$, and $b$.

Having done the unary and binary case of sums of sets, we might as well do the nullary one, too. 
The System F style encoding of the empty type  $\emptyt$ in $\UU$, given by
\[
\emptyt\;\defeq\;\tprd{X:\UU}X\,,
\]
admits a recursor 
\[\rece:\equiv\,\lam {Xc}c_X\,:\,\itprd{X:\UU}\emptyt\to X\,,\] 
and it turns out that in this case we don't need any refinement, since we can already derive that $\emptyt$ is a proposition and that $\rece$ satisfies an $\eta$ equality -- indeed, we have 
\[
\lam{c d}c (c=d):\prod_{c,d:\emptyt}(c=d)\ \equiv\,\isprop(\emptyt)\,,
\]
and the eliminator $\rece$ is unique, by
\[
\funext(\lam{c}c(c_X=fc))\,:\,{\rece}X =f
\]
for any $X:\UU$ and $f : \emptyt\to X$.


\paragraph*{Other non-recursive 0-types}

Unlike $\emptyt$, the terminal set $\unit$ does not have the System F form $\unit_F :\equiv\,\prod_{X:\UU}X\to X$ (uniqueness of the maps $X\to \unit_F$ fails), but instead can be encoded as a (-1)-type via the familiar
\begin{equation}\label{eq:true}
\unit\,:\equiv\, \prod_{p:\Prop} p\to p\,.
\end{equation}
Indeed, this is easily seen to be terminal for all $X:\UU$.

The method of adding a naturality condition can be used to encode the \emph{set-truncation} $\truncation{A}_0$ 
of an arbitrary type $A$. Indeed, we can simply take
\begin{equation}\label{eq:settrunc}
\begin{split}
\truncation{A}_0\ &:\equiv \sum_{\alpha : A^*} \prod_{X,Y:\set} \prod_{f:X\rightarrow Y} \alpha_Y \circ f^A = f\circ \alpha_X\,,\\
\text{where}\quad A^*\,&:\equiv\,\prod_{X:\set} (A\rightarrow X)\rightarrow X\,,
\end{split}
\end{equation}
as in the Basic Lemma \ref{lem:basic}, since $A^*$ is a set even when $A$ is not one.

Observe how this generalizes the well-known \cite{hofmann1995extensional,Awodey:2004:PT:1094469.1094484} \cite[Exercise~3.15]{HoTTbook1} 
\emph{propositional truncation} of a type $A$,
\[
\truncation{A} =\ \prod_{X:\prop}\, (A\rightarrow X)\rightarrow X\,.
\]

We only mention that it is also possible to give correct impredicative encodings of set-quotients \cite[6.10]{HoTTbook1}, as well as general coequalizers of sets, by related methods.

\section{General inductive sets}\label{sec:generalz}

While sums and truncations are viewed as inductive types in modern terminology, the classical idea of
an inductive type involves generation from constants by repeated application of constructors. A well 
understood class of inductive types in type theory and category theory are \emph{W-types}~\cite[pg.~43]{martin-lof:bibliopolis}, which are generated from a family of constructors of specified -- possibly infinite -- arities. 

Inductive types of this kind with only \emph{finitely}
many constructors, each of finite arity, can be encoded in System F~\cite[Section 11.5]{girard:proofsandtypes}, and translating these encodings into type theory by quantifying over the impredicative universe $\set$ leads again to types which admit the correct constructors and recursors but fail to satisfy the appropriate $\eta$-rules.

On the other hand, it is known from category theory that W-types can be understood as initial algebras
of so called \emph{polynomial functors}~\cite{MoerdijkPalmgren2002,gambino2003wellfounded}.
%
In the following we show how a categorical construction of initial algebras relying on the impredicativity of $\set$ gives rise to subtypes of the System F style encodings satisfying $\eta$.
We elaborate this idea using as running example the inductive type of natural numbers, but the method is easily seen to generalize.
%
%

\subsection{Initial algebras of endofunctors}\label{sec:inalg}

As pointed out in Section~\ref{sec_system}, the precategory $\catset$ is complete in the very strong sense that it has all equalizers (constructed using $\Sigma$- and  identity-types) as well as products of families of objects indexed by \emph{arbitrary} types. It is an old observation by Hyland~\cite[Section~3.1]{HYLAND1988135} that this implies the existence of initial algebras for arbitrary endofunctors, at least in the related semantical setting of certain kinds of internal categories.
In the following we give explicit type-theoretic descriptions of the required limits and initial algebras by unwinding the categorical definitions. 

First, observe that limits over arbitrary (pre)category-indexed diagrams can be expressed using products and equalizers, as usual.  Specifically, let $\catj$ be a precategory: an arbitrary type of objects $\catj_0$ and a family of \emph{sets} of arrows $\hom:\catj_0\times\catj_0\to\set$, equipped with the usual composition and unit structure, and satisfying the usual equations on these, which are \emph{propositional}, because the $\hom(i,j)$ are sets for all $i,j:\catj_0$.  A $\catj$-indexed `diagram' is just a functor $D:\catj\to\catset$, which can also be defined as usual, since the values $D_i$ are all sets.
 The limit of $D$ is the equalizer of the two maps
\begin{align*}
p,q\;\;:\;\;&\prd{i:\catj_0}D_i\;\rightrightarrows\;\prd{i,j:\catj_0}\prd{u:\hom(i,j)} D_j\\
p\;\phi\; i\; j\; u \;&:\equiv\; D(u)(\phi_i)\\
q\;\phi\; i\; j\; u \;&:\equiv\; \phi_j\,, 
\end{align*} 
which is given explicitly by the type
\begin{equation}\label{eq:lim}
\begin{split}
\varprojlim_{i} D_i\;\defeq\;&\sm{\phi:D^*}\prd{i,j : \catj_0}\prd{u:\hom(i,j)}\;D_u(\phi_i) = \phi_j,\\
\text{where}\quad D^*\;\defeq\;&\prd{i:\catj_0}D_i\,,
\end{split}
\end{equation}
together with projections
\[
\pi_j\,\defeq\; \lam{\xi}(\fst\xi)_j\,:\,\varprojlim_{i} D_i\to D_j\quad\text{for } j:\catj_0.
\]
Crucially for proving the $\eta$ rule in \Cref{sec:nat}, observe that $\varprojlim_{i} D_i$ is a set, since all the $D_i$ are sets, $\set$ is  impredicative, and $\varprojlim_{i} D_i$ is therefore a sum of a family of propositions over a set.

Now recall that, given an endofunctor $F:\catset\to\catset$, the category $\falg$ of \emph{$F$-algebras} has  as objects pairs $(X:\set,\alpha:FX\to X)$, and as morphisms from $(X,\alpha)$ to $(Y,\beta)$ the functions $f:X\to Y$ satisfying $f\circ\alpha=\beta\circ Ff$.  Thus, type-theoretically, we have the precategory:
\begin{align*}
(\falg)_0 \;&\defeq\; \sm{X:\set}FX\to X\,,\\
\hom\big(\,(X,\alpha),\, (Y,\beta)\,\big) \;&\defeq\; \sm{f:X\to Y}f\circ\alpha=\beta\circ Ff\,.
\end{align*}
The forgetful functor $U:\falg\to\catset$ is just the first projection, that is: 
\begin{align*}
U(X,\alpha)\ &\defeq\ X &&\text{for}\quad(X,\alpha):(\falg)_0\,,\\
\qquad U(f,p)\ &\defeq\ f&&\text{for}\quad(f,p):\hom\big(\,(X,\alpha),\, (Y,\beta)\,\big).
\end{align*}
It is well-known, and easy to prove, that the precategory $\falg$ inherits arbitrary limits from $\catset$, and these can be computed pointwise.  Thus $\falg$ has an initial object, which is just the limit of the identity functor. Since limits in $\falg$ are computed pointwise, they are preserved by $U$, which means that the initial algebra $(I,\, i : FI\to I)$ has as its underlying set $I$ the limit of the functor $U:\falg\to\catset$, which using \eqref{eq:lim} we can write explicitly as
\begin{align}\label{eq:F-algtype}
I\,\defeq\;&\sm{\phi:U^*}\mathsf{Lim}(\phi)
\end{align}
\begin{align*}
\text{where}\quad U^*\;\defeq\;&\prd{A:\falg}UA\,,\\
\text{and}\quad \mathsf{Lim}(\phi)\,\defeq\;&\prod_{(A,B:\falg)}\prod_{(f,p):\hom(A,B)}f(\phi_A)=\phi_B\,.
\end{align*}
Summarizing the foregoing discussion, we have the following.
\begin{theorem}\label{thm:initial}
For any functor $F : \catset \to \catset$, the category $\falg$ of $F$-algebras has an initial object
\[
i:FI \rightarrow I,
\]
where:
\begin{enumerate}
\item the set $I$ is given by the type \eqref{eq:F-algtype} and is the limit of the forgetful functor $\,U:\falg\to\catset$, and
\item the map $i: FI\to I$ is given by
\[
i(x)\,\defeq\,(\lam{A}(\snd A)(F(\pi_A)x),\,q(x))\,,
\]
where $q : \mathsf{Lim}(\lam{A}(\snd A)(F(\pi_A))$ is constructed from functoriality of $F$ and naturality of the limit cone.
\hfill$_\square$
\end{enumerate} 
\end{theorem}

We emphasize that the foregoing theorem is not merely semantically true in a certain model, but is provable in our system of type theory.  
In the following we use this construction of initial algebras to obtain an encoding of the type of natural numbers which refines the System F encoding.

\subsection{Natural numbers}\label{sec:nat}

The inductive type $\N$ of natural numbers is generated by the constructors
\[ 
0:\N\qquad\qquad\suc:\N\to\N.
\]
From this specification we can derive the set-level System F style encoding
\begin{equation}\label{eq:NatF}
\natf\ \defeq\ \tprd{X:\set} (X\rightarrow X)\rightarrow X\rightarrow X
\end{equation}
admitting constructors
\begin{align*}
	\zerof &\defeq \lam{\,X\,h\,x}x\;:\;\natf\\
	\sucf &\defeq\lam{\,n\,X\,h\,x}\;h\,(n\,X\,h\,x)\,:\,\natf\to\natf
\end{align*}
and a recursor
\begin{align*}
\recnf \defeq&\;\lam{\,X\,h\,x\,n}nXhx\;:\;\tprd{X:\set}(X\to X)\to X\to\natf\to X.
\end{align*}
These satisfy the $\beta$-rules
\[\recnf\,h\,x\,\zerof \equiv x \qtext{and} 
\recnf\,h\,x\,(\sucf\,n) \equiv h(\recnf\,h\,x)\]
for $X:\set$, $h:X\to X$ and $x:X$, but not the \emph{$\eta$-rule},
which states that $\recnf\,h\,x$ is uniquely determined in the sense that 
\[
\big(\, f(\zerof)=x\,\wedge\, f\circ\sucf=h\circ f\,\big) \to\, \big(\,f\,=\,\recnf \,h\,x\,\big)
\]
for all $f:\natf\to X$ (and writing $\wedge$ in place of $\times$ for propositions).

On the other hand, $\N$ can be categorically characterized as the initial algebra of the functor
%
%
%
functor $T:\uset\to\uset$ given by
\begin{equation*}
\begin{split}
T(X)&\defeq X+\unit\,,\\
T(f)&\defeq \recp\,(\inl\circ f)\,\inr\quad\text{for }f:X\to Y\,.
\end{split}
\end{equation*}
Here $\unit$ is the unit type from \eqref{eq:true}.

Instantiating the type in Theorem~\ref{thm:initial} we get the type
%
%
\[
\tsm{\phi: \tprd{A:\tsm{X:\set}X+\unit\to X}\,\fst A}\,
\mathsf{Lim}(\phi)
\]
for the underlying set of the initial algebra, and it turns out that the index type of the sum
is equivalent to $\natf$:
\begin{align*}
\prod_{A:\sum_{X:\set}X+\unit\to X}\,\fst A
&\simeq\, \prod_{X:\set}(X+\unit\to X)\to X\\
&\simeq\, \prod_{X:\set}(X\to X)\times X\to X\\
& \simeq\;\natf
\end{align*}
We compose the predicate $\mathsf{Lim}$ with this equivalence to get a description of the initial algebra directly as a subtype $\mathsf{Lim}'$ of the System F encoding:
\begin{align}
\N\,\defeq\,&\tsm{\nu:\natf}\,\mathsf{Lim}'(\nu)\quad\text{with}\nonumber\\[-3mm]
\mathsf{Lim}'(\nu)\,\defeq\,&\prod_{\substack{X:\set\\Y:\set}}\,\prod_{\substack{x:X\\y:Y}}\,\prod_{\substack{h:X\to X\\k:Y\to Y\\ 
		f: X\to Y}}
\begin{array}{l}
\\[2mm]
(f(x)=y\wedge f\circ h=k\circ f)\\
\quad\to f(\nu_Xh\,x)=\nu_Yk\,y
\end{array}\label{eq:mprime}
\end{align}
Observe that the triples $(X,h,x)$ and $(Y,k,y)$ in the definition of $\mathsf{Lim}'(\nu)$ can be coerced into $T$-algebras $(X,\,\recp h\,(\lambda z.\, x))$ and $(Y,\,\recp k\,(\lambda z.\,y))$. Leaving this coercion implicit (as we shall do from now on), $\mathsf{Lim'}(v)$ can be read as saying that we have \[f(\nu_Xh\,x)=\nu_Yk\,y\] for every $T$-algebra morphism 
\[f:(X,h,x)\to(Y,k,y). \]

It is easy to see that the predicate $\mathsf{Lim}'$ satisfies
\[
\mathsf{Lim}'(\zerof)\qtext{and} \tprd{\nu:\natf}\, \mathsf{Lim}'(\nu)\to \mathsf{Lim}'(\sucf \nu),
\]
whence the System F constructors $\zerof$ and $\sucf$ can be restricted to operations
\[
0:\N\qtext{and}\suc:\N\to\N
\]
about which it is sufficient to know that they behave like $\zerof$ and $\sucf$ on the first components of dependent pairs.
The recursor $\recnf$ also restricts to $\N$ in a straightforward manner
\begin{align*}
\recn \defeq\;\lam{\,X\,h\,x\,n}&\fst(n)X\,h\,x\\&:\,\itprd{X:\set}(X\to X)\to X\to\N\to X\,,
\end{align*}
%
%
%
and we have the following theorem.
\begin{theorem} The encodings of $\,\N$, $0$, and $\suc$ given above satisfy
	\begin{enumerate}
\item\label{natthm2'} \emph{definitional $\beta$-rules} saying that
\[\recn\,h\,x\,0 \equiv x \qtext{and} \recn\,h\,x\,(\suc\,n) \equiv h\,(\recn\,h\,x)\]
for all $X:\set$, $x:X$, $h:X\to X$ and $n:\N$, and 
\item\label{natthm4'} a \emph{propositional $\eta$-rule} which states that 
\[
\big(\,f(0)=x\wedge f\!\circ\!\suc=h\circ f\,\big)\ \to\ f=\recn \,h\,x
\]
for all $X:\set$, $x:X$, $h:X\to X$ and $f:\N\to X$.
\end{enumerate}
\end{theorem}
\begin{proof}
The first claim is straightforward.

For the second claim we first show that 
\begin{equation}\label{eq:wneta}
\recn\,\suc\,0 = \mathrm{id}_{\N}.
\end{equation}
By function extensionality and $\Sigma$-induction it is enough to show that
$n\,\N\,\suc\,0=(n,p)$ for all $n\in\natf$ and $p\in \mathsf{Lim}'(n)$. Since $\natf\emb\N$ is an embedding this reduces to $\fst(n\,\N\suc\,0)=n$, and by function extensionality again it suffices to show that
\[\fst(n_\N\,\suc\,0)_Xf\,x=n_Xf\,x
\]
for all $X:\set$, $f:X\to X$ and $x:X$. Since $\recn\,f\,x:\N\to X$ is a morphism of $T$-algebras
$(N,\suc,0)\to(X,f,x)$, $p$ gives an equality between the right hand side and $\recn\,f\,x(n_\N\,\suc\,0)$ (see remark after~\eqref{eq:mprime}), which in turn is definitionally equal to the left hand side.

Now let $X:\set$, $x:X$, $h:X\to X$ and $f:\N\to X$ such that $f(0)=x$ and $f\circ\suc=h\circ f$.
Given $n:\N$ we argue 
\begin{align*}
f\,n\,&=f(\recn\,\suc\,0\,n)&&\text{(by }\eqref{eq:wneta})\\
&\equiv f(\fst(n)_\N\suc\,0)\\
&= \fst(n)_Xh\,x&&
\begin{array}{@{}l}
\text{(since $f$ is a $T$-algebra morphism}\\
\text{from $(N,\suc,0)$ to $(X,h,x)$)}
\end{array}
\\
&\equiv\recn\,h\,x\,n,
\end{align*}
which by function extensionality proves the claim.
\end{proof}
Apart from the distinction of definitional and propositional equality, the preceding theorem says precisely that the $F$-algebra $(\N,\suc,0)$ is initial in $\talg$, where for any $T$-algebra $(C,f,c)$ the function underlying unique mediating morphism \[(\N,\suc,0)\to(C,f,c)\] is given by $\recn\,f\,c$.
We reiterate that this assertion is to be understood as a statement \emph{in type theory} (as in \cite{ags:it-hott,AGS2}), in particular 
the uniqueness of the mediating morphism is up to propositional equality.
%

As before, the \emph{induction principle} for \N follows from recursion together with the $\eta$-rule (by \emph{ibid.}).

\subsection{Consequences of the existence of initial algebras}

We note that the sets in this impredicative system are necessarily quite non-classical -- as has been observed in related systems by several previous authors \cite{girardoriginal,10.1007/3-540-13346-1_7, PittsAM:polist, HYLAND1988135}.

\begin{corollary} The initial algebra theorem \ref{thm:initial} implies the following facts about the category $\uset$:
\begin{enumerate}
\item Every endofunctor $F : \uset \to \uset$ has a (least) fixed point, 
$$F(X) \cong X.$$
(By Lambek's lemma.)

\item There is no (covariant) powerset functor $P : \uset \to \uset$.\\
(Otherwise we would have $P(X)\cong X$ for some set $X$.)

\item The law of excluded middle fails for sets.\\
(Otherwise we would have $P(X)=2^X$.)
\end{enumerate}
\end{corollary}

\paragraph*{Other 0-types} We also remark that, using Theorem \ref{thm:initial}, one can encode initial algebras for other \emph{polynomial} endofunctors on $\uset$ in just the same way, and thus obtain all (set-level) W-types (see \cite{ags:it-hott,AGS2}). It follows that the internal category $\uset$ of sets is a (complete!) \emph{predicative topos} in roughly the sense of \cite{2012arXiv1207.0959V}, i.e.\ an LCC pretopos with W-types.

\section{Some 1-Types}\label{sec:1types}

The foregoing development of impredicative encodings of inductive \emph{sets} is satisfactory as such, but in the full system of HoTT one  also has higher n-types, and for these there are corresponding notions of inductive type, known as \emph{higher inductive types} (HITs).  Examples include some basic spaces such as the spheres $S^n$, homotopy colimits, $n$-truncations $\trunc{n}{X}$, and many others; see \cite[ch.\ 6]{HoTTbook1}.\footnote{Strictly speaking, some $\set$-level inductive types, such as quotients, are also HITs, in that they involve primitive identity paths.  Even the propositional truncation is a HIT in this sense.}

In this section, we give an example of an impredicative encoding of a basic ``1-HIT'', namely the 1-sphere $\mathsf{S}^1$ \cite{ls:pi1s1}. We give fewer details of the (more intricate) proofs, but  provide enough specifics to hopefully give the reader a sense of what is involved in such higher encodings.  

The 1-sphere $\mathsf{S}^1$ is defined as a HIT by the constructors
\begin{equation*}
\begin{aligned}
	\mathsf{base} &: \mathsf{S}^1 \\
	\mathsf{loop} &: \mathsf{base} = \mathsf{base}\,.
\end{aligned}
\end{equation*}
Its dependent eliminator, given in \cite[6.4]{HoTTbook1},
can be derived from its recursor
\begin{equation*}
	\rec{\mathsf{S}^1}\, :\, \prod_{X:\typele{1}} \prod_{x:X}\,(x=x)\to \mathsf{S}^1 \to X\,,
\end{equation*}
together with the $\beta$-rules
\begin{equation*}
\begin{aligned}
	\rec{\mathsf{S}^1}\mathsf{base} &\equiv x\,,\\
	\rec{\mathsf{S}^1}(\mathsf{loop})
    &= p\,,
\end{aligned}
\end{equation*}
for $X:\typele{1}$, $x:X$ and $p : x=x$,
and a propositional $\eta$-rule stating the uniqueness of the recursor (see \cite{Sojakova2015HigherIT}). 

We begin with the following encoding (originally proposed by Mike Shulman \cite{shulmanblog}), which is suggested by the previous `System~F~style' ones.
\begin{equation}\label{eq:S1}
	\mathsf{S}^1_F\ =\  \prod_{X:\UU}\prod_{x:X}\, (x=x)\rightarrow X\,.
\end{equation}
This has the same problem as the System F encoding of $\mathbb{N}$ \eqref{eq:NatF}, however: no uniqueness for the recursor, 
and so no induction principle.  We will remedy this in the same way as before, now restricting the $\prod_{X}$ to 1-types, and then adding naturality and a higher ``coherence condition'', reflecting the fact that $\mathsf{S}^1$ is a 1-type rather than a set.

First, to see where \eqref{eq:S1} came from, the \emph{universal property of the circle} \cite[Lemma 6.2.9]{HoTTbook1} is given by the equivalence,
\begin{equation}\label{eq:UMPS1}
(\mathsf{S}^1 \rightarrow X)\  \simeq\ \sum_{x:X}x=x\,.
\end{equation}
That is to say, maps $\mathsf{S}^1 \rightarrow X$ correspond to ``loops'' $x=x$ in $X$, with various basepoints $x:X$.  

Thus by the same reasoning as in the Basic Lemma, if we had $\mathsf{S}^1$, there would be an embedding,
\begin{align*}
	\mathsf{S}^1\ &\emb\ \prod_{X: \typele{1}}\, (\mathsf{S}^1 \rightarrow X) \rightarrow X\\
			&\simeq\  \prod_{X: \typele{1}}\, (\sum_{x:X}x=x) \rightarrow X\qquad \qquad(\text{by }\eqref{eq:UMPS1})\\
			&\simeq\  \prod_{X: \typele{1}}\prod_{x:X}\, (x=x) \rightarrow X\,.
\end{align*}
This gives us \eqref{eq:S1} as a starting point, but we again need to refine the encoding to a suitable subtype.    

Let us write 
\[
\Omega(X)\ :\equiv\ \sum_{x:X}x=x
\]
 for the (unbased) ``loopspace functor'', 
 so that \eqref{eq:S1} is essentially the type of ``families of maps $\Omega{X} \to X$''.  Let $\alpha : \prod_{X:\typele{1}} \Omega{X} \rightarrow X$
be such a family, and consider the following naturality square, corresponding to \eqref{eq:naturality}
\begin{equation*}\label{eq:S1naturality}
\xymatrix{
\Omega{X} \ar[d]_{\alpha_X} \ar[r]^{\Omega{f}} &   \ar[d]^{\alpha_Y}  \Omega{Y}  \\
X \ar[r]_{f} & Y\,,
}
\end{equation*}
where $\Omega{f}$ is defined in the usual way as the action of $f$ on identity paths in $X$ \cite[2.2]{HoTTbook1}.  Since we now have a 1-type $Y$ in the target, the naturality condition 
\begin{equation*}\label{eq:encodeS1}
f\circ \alpha_X\  =\ \alpha_Y \circ \Omega{f}
\end{equation*}
is not a proposition, but a set.  We will add a ``coherence condition'' in order to cut it down further.  
Indeed, consider a naturality term 
\[
\vartheta : \prod_{X,Y:\typele{1}} \prod_{f:X\rightarrow Y} f\circ \alpha_X\  =\ \alpha_Y \circ \Omega{f}\,. 
\]
Then for any $f:X\rightarrow Y$ and $g:Y\rightarrow Z$, we have identifications 
\begin{align*}
\vartheta_f :&\ f\circ \alpha_X\, =\, \alpha_Y \circ \Omega{f}\,,\\
\vartheta_g :&\ g\circ \alpha_Y\,=\,\alpha_Z \circ \Omega{g}\,,
\end{align*}
depicted
\begin{equation*}
\xymatrix{
\Omega{X} \ar[d]_{\alpha_X} \ar[r]^{\Omega{f}} \ar@{}[rd]|{\vartheta_f} 
	& \ar[d]|{\alpha_Y} \Omega{Y}  \ar[r]^{\Omega{g}} \ar@{}[rd]|{\vartheta_g} 
	& \ar[d]^{\alpha_Z} \Omega{Z} \\
X \ar[r]_{f} & Y  \ar[r]_{g} & Z\,,
}
\end{equation*}
as well as one
\begin{align*}
\vartheta_{(g\circ f)} :&\ (g\circ f)\circ \alpha_X\  =\  \alpha_Z \circ \Omega{(g\circ f)}\,.
\end{align*}
Informally, the coherence condition we seek is then 
\[
\text{``}\ \vartheta_{(g\circ f)} = \vartheta_g \cdot \vartheta_f \ \text{''}
\]
but of course, this does not type-check. Instead, we have the following ``pasting scheme'' 
\begin{equation*}
\xymatrix{
\Omega{X} \ar[d]_{\alpha_X} \ar[r]^{\Omega{f}} \ar@{}[rd]|{\vartheta_f}  
			\ar@/^2pc/ [rr]^{\Omega{(g\circ f)}} \ar@{{}{ }{}}@/^1pc/ [rr]|{\varphi_{f,g}} 
		& \ar[d]|{\alpha_Y} \Omega{Y}   \ar[r]^{\Omega{g}} \ar@{}[rd]|{\vartheta_g}  & \ar[d]^{\alpha_Z} \Omega{Z}  \\
X \ar[r]_{f} \ar@/_2pc/ [rr]_{g\circ f} & Y  \ar[r]_{g} & Z\,,
}
\end{equation*} 
where $\varphi_{f,g} : \Omega g \circ\Omega f\, =\, \Omega(g\circ f) $ is the provable composition law of the (pseudo-) functor $\Omega$.  
This gives rise to the well-typed \emph{coherence condition},
\begin{equation}\label{eq:coherence}
\vartheta_{(g\circ f)} = (\alpha_Z *\varphi_{f,g})\cdot (\vartheta_g * \Omega{f}) \cdot (g * \vartheta_f) \,,
\end{equation} 
where $q\cdot p$ is composition of the paths $p, q$, and $f*p$ (resp.\ $p*f$) is the ``whiskering'' of a map $f$ and a path $p$ (see \cite[2.1]{HoTTbook1}). 

Note that \eqref{eq:coherence} is indeed a proposition, because it is an identity between identities $(g\circ f)\circ\alpha_X\, =\, \alpha_Z\circ\Omega(g\circ f)$ in the 1-type $\Omega{X} \to Z$.

There is also a \emph{unit coherence condition}, which has the simple form
\[
\vartheta_{(1_X)} = \refl{\alpha_X}\,
\]
in the 1-type $\Omega{X} \to X$ (taking into account the fact that ``$\Omega$ preserves $\refl{}$'').

Now for $\alpha:\prod_{X:\typele{1}} \Omega{X}\rightarrow X$, let us write
\[
\mathrm{Nat}(\alpha)\ :\equiv\ \prod_{X,Y:\typele{1}} \prod_{f:X\rightarrow Y} f\circ \alpha_X = \alpha_Y \circ \Omega{f}
\]
for the type of ``naturality structures'' on the family of maps~$\alpha$, and for $\vartheta : \mathrm{Nat}(\alpha)$, let us write
\begin{align*}
\mathrm{Coh}(\vartheta)\,:\equiv\,
	&\prod_{X,Y, Z:\typele{1}}
    \prod_{f:X\rightarrow Y} 
    \prod_{g : Y\rightarrow Z}\,  \big(\vartheta_{(1_X)} = \refl{\alpha_X}\big)\\
	&\quad\times\, \big(\vartheta_{(g\circ f)} = (\alpha_Z *\varphi_{f,g})\cdot (\vartheta_g * \Omega{f}) \cdot (g * \vartheta_f)\big)
\end{align*}
for the type of ``coherence conditions'' on the natural transformation~$\vartheta$.

The sharper encoding of $\mathsf{S}^1$ that we seek is then
\begin{equation}\label{eq:finalencodeS1}
	\mathsf{S}^1\, :\equiv \sum_{\alpha : (\mathsf{S}^1)^*}
	\sum_{\vartheta: \mathrm{Nat}(\alpha)} \, 	\mathrm{Coh}(\vartheta)
\end{equation}
where, as before, 
\begin{align*}
(\mathsf{S}^1)^* 
&:\equiv \prod_{X:\typele{1}} \prod_{x:X}(x=x)\rightarrow X\\
&\simeq\ \prod_{X:\typele{1}} (\sum_{x:X} x=x )\rightarrow X\\
&\equiv\ \prod_{X:\typele{1}} \Omega{X}\rightarrow X\,.
\end{align*}
The type \eqref{eq:finalencodeS1} can thus be understood as consisting of those natural transformations $\Omega{X} \to X$ that are coherent (also called \emph{pseudo-natural transformations}).

The constructors $\mathsf{base}$ and $\mathsf{loop}$ can then be defined:
\begin{equation*}\label{eq:S1inj}
\begin{aligned}
	\mathsf{base}^* &:\equiv\ \lam{Xx\,p}\, x  \\
		&: \prod_{X:\typele{1}} \prod_{x:X}(x=x)\rightarrow X\\
	\mathsf{base}	&:\equiv\ \big(\,\mathsf{base}^*,\, (\,\lambda XYf.\,\refl{},\, \lambda XYZfg.\,(\,\refl{},\refl{}\,)\,)\,\big) \\
		&: \ \sum_{\alpha : (\mathsf{S}^1)^*} \sum_{\vartheta: \mathrm{Nat}(\alpha)} \, \mathrm{Coh}(\vartheta)\\
	\mathsf{loop}^* &:\equiv\ \funext(\lam{X\,x\,p}\, p) \\
		&: \mathsf{base}^* = \mathsf{base}^* \\
	\mathsf{loop} &:\equiv\ \big(\mathsf{loop}^*, (a, (b, c))\big)\\ 
		&: \ \mathsf{base} = \mathsf{base}\,.
\end{aligned}
\end{equation*}
where the subterms $a, b, c$ in $\mathsf{loop}$ are certain canonical higher coherences, such as $\mathsf{inv} : \vartheta\cdot\vartheta^{-1} = \refl{}$, the details of which which we omit.
Finally, the recursor $\mathsf{rec}_{\mathsf{S}^1}$ is given by:
\begin{equation*}\label{eq:S1rec}
\begin{aligned}
\mathsf{rec}_{\mathsf{S}^1} &:\equiv\ \lam{\,Cz\,p\,\sigma}(\fst\,\fst\,\sigma)_C\,z\,p \\
&: \prod_{C:\typele{1}}\prod_{z:C} (z=z)\to \mathsf{S}^1\to C
\end{aligned}
\end{equation*}

\begin{theorem}
The encoding \eqref{eq:finalencodeS1} of $\mathsf{S}^1$, with the structure $\mathsf{base}$ and $\mathsf{loop}$ just given, satifies the rules for the corresponding higher inductive type, including the dependent elimination rule, with respect to (families of) 1-types.
\end{theorem}

For the proof, as in previous cases one shows first that the encoding implies the recursion rules for $\mathsf{S}^1$, including the $\eta$-rule.  Then, by the results in \cite{sojakova-thesis,Sojakova2015HigherIT}, it also admits dependent elimination.

Observe that our encoding can be motivated informally by a version of the Basic Lemma \ref{lem:basic} for 1-types (i.e.\ the Yoneda lemma for bicategories) as follows, where we write $I$ for the identity functor on the (bi)category of 1-types, and $\mathbf{pNat}$ for the groupoid of pseudo-natural transformations between pseudo-functors,
\begin{align*}
\mathsf{S}^1 \equiv I(\mathsf{S}^1)\ 
&\simeq\ \mathbf{pNat}(\,(-)^{\mathsf{S}^1},\, I\,)&&\text{by Yoneda}\\
&\simeq\ \mathbf{pNat}(\,\Omega\,,\, I\,)&&\text{by \eqref{eq:UMPS1}}\\
&\equiv\ \sum_{\alpha : (\mathsf{S}^1)^*} 
\sum_{\vartheta: \mathrm{Nat}(\alpha)}
\mathrm{Coh}(\vartheta)\,. &&
\end{align*}

\paragraph*{Other 1-types}

Similar methods can be used to give the 1-truncation $\truncation{A}_1$ of a general type $A$, in the analogous form

\begin{equation*}\label{eq:1truncation}
\trunc{1}{A}\, :\equiv\, \sum_{\alpha : A^*}
\sum_{\vartheta: \mathrm{Nat}(\alpha)} 	\mathrm{Coh}(\vartheta)
\end{equation*}
where, as in \eqref{eq:settrunc}, 
\begin{align*}
A^*\, &:\equiv \prod_{X:\typele{1}} (A\to X)\rightarrow X\,.
\end{align*}

By analogy to the encoding of set quotients, we expect that one can also use the foregoing method to encode the groupoid quotient (\cite{aks:rezk}: ``Rezk completion'') $BG$ of a (pre)groupoid $G$ (cf.~\cite{Sojakova2015HigherIT} for a type theoretic description). 
In principle, one should also be able to encode some ``higher W-types'', specified by families $x:A \vdash B(x)$ of 1-types $B(x)$, such as occur in the theory of combinatorial species \cite{JOYAL19811}.  But we have not investigated these possibilities.

\section{Conclusion}\label{sec:limits}

\paragraph*{Limitations} 
The methods used here have certain limitations which we can now address. The first is the familiar (in type theory) issue of ``large versus small elimination'', i.e.\ elimination into types in higher universes. The encodings given here have elimination rules only with respect to types in the impredicative universe $\UU$.  That said, we made little use of the higher universes $\UU_0:\UU_1:\UU_2: ...$ and could simply have omitted them (with some corresponding adjustments). The resulting system would be similar to the original Calculus of Constructions  \cite{COQUAND198895}, but with the addition of primitive identity-types.  The issue of eliminating into $\UU$ itself of course remains, as it does for the CoC.

Restricting the system to the subuniverse $\set$ of 0-types -- either by definitions as was done here, or by the addition of ``extensionality'' principles such as the uniqueness of identity proofs as in \cite{altenkirch+07ott} or even the reflection rule as in \cite{Dybjer:1997} -- seems to give a very satisfactory type theory for sets, with quotient types, W-types, etc., within a framework with very few primitive operations.    

In an intensional system admitting higher $n$-types, there is the interesting, and apparently new, question of whether encoded $n$-types can eliminate into ($n+k$)-types (note that the truncation levels are cumulative, so the converse is immediate).  
It seems that this is possible in some cases, but  more work needs to be done to understand which ones. 
 To give a brief indication, we have already seen the encodings $\emptyt\defeq\prod_{X:\UU}X$
and $ \unit\defeq\prod_{p:\Prop} p\to p$\,, but why does the first one require a product over all of $\UU$, while $\Prop$ suffices for the second?  The issue seems to be related to the preservation of certain (co)limits by the inclusions $\typele{n} \emb\typele{n+1}$ (of course, the truncations are left adjoints to these).  It seems that some encodings are ``stable'' in the sense that they also eliminate into higher $n$-types; for example, our encoding of $\N$ as a set should be stable, as can be seen by considering the fact that it is a W-type \cite{AGS2}, constructed in $\set$, and the truncation levels are closed under W-types \cite{danielssonW}.  
Similar considerations apply to other inductive 0-types encoded by the method in \Cref{sec:generalz}.  
Indeed, we could also encode the natural numbers $\N$ as a 1-type, but then show that the result is actually a 0-type and thus equivalent to the original encoding, which therefore now eliminates into all 1-types.  In this way, it seems possible to establish that some $n$-types can indeed eliminate into ($n+k$)-types.  
A similar argument would seem to apply to the encoding of  $\mathsf{S}^1$ given in \Cref{sec:1types}, which should therefore also be stable. An encoding that is apparently \emph{not} stable would be a coequalizer constructed in $\set$; constructing it in a higher $\typele{n}$ should give a homotopy colimit, which need not be a set.

The previous considerations are particularly tentative, however, because of the combinatorial difficulty of specifying the relevant coherence conditions for higher $n$-types.   Indeed, the current methods are restricted to encodings of $n$-types only for very low $n$.  More work, and perhaps some new insight, is needed to specify the coherence conditions for higher $n$-types, and for untruncated types, such as the $2$-sphere $\mathsf{S}^2$. 

We also recognize that, as  experience has shown,  reasoning effectively about (higher) inductive types often involves large eliminations into a (univalent) universe.
For conventional inductive types, the lack of such large elimination means that we can not compute initial algebras by iterating functors (since this would require elimination from $\N$ to $\UU$), but fortunately we don't need this kind of iteration in the impredicative setting, since we can compute initial algebras differently.

For higher inductive types the situation may be different, however; large elimination (plus univalence) seems to be essential
in the proof that $\pi_1(\mathsf{S}^1) = \Z$ (see \cite{Licata:2013:CFG:2591370.2591407}).  Thus, ultimately, the utility of such encodings of HITs may be that they serve to justify adding the associated rules globally, as was done for conventional inductive types in passing from the Calculus of Constructions to the Calculus of Inductive Constructions \cite{CoquandPaulin,PfenningPaulinMohring,Dybjer:1991}.

\paragraph*{Semantics}

A realizability model of the calculus of constructions with a proof-relevant impredicative universe is described in~\cite[Chapter~2]{Streicher-1991}.
The impredicative universe in this model is closed under small sums and identity types, and the model can be extended to include a hierarchy of predicative universes using techniques akin to those discussed in~\cite{streicher2005universes} (assuming sufficiently many set-theoretic Grothendieck universes). This establishes the consistency of our system and provides a semantic framework for impredicative encodings of $0$-types. Non-trivial 1-types can be modeled in a \emph{groupoidal realizability model}, obtained by internalizing the Hofmann-Streicher groupoid model~\cite{hs:gpd-typethy} in a realizability topos. A model of a type theory with an impredicative universe and types of arbitrary $h$-level is conjectured in a putative `realizability $\infty$-topos'. This is work in progress.

\paragraph*{Related and future work}

Most of the results presented here are treated in more detail in the third-named author's M.S. thesis \cite{speight-masters-thesis}. 

A formalization of the main results contained herein is currently underway using (an impredicative branch of) the Lean proof assistant.  The files are publicly available here \cite{formalization}.

\section*{Acknowledgements}
This material is based upon work supported by the Air Force Office of Scientific Research under MURI Grant No.~FA9550-15-1-0053.  Any opinions, findings and conclusions or recommendations expressed in this material are those of the authors and do not necessarily reflect the views of the AFOSR.
  
Thanks to Andrej Bauer and Mike Shulman for discussions of earlier versions of this work.

\bibliographystyle{alpha}
\bibliography{references.bib}

\appendix
\section{System F}\label{appendix_F}
Starting with countably many type and term variables $X,Y,Z,...$ and $x,y,z,...$ respectively, System F types and untyped terms are generated by the following BNF grammars.
\begin{align*}
A,B &::= X \;|\; A\rightarrow B \;|\; \forall X.A\\
t,u &::= x \;|\; \lambda(x:A).t \;|\; tu \;|\; \Lambda X.t \;|\; tA
\end{align*}
Typed terms in context are derived via the following rules.
		\begin{mathparpagebreakable}
			\def\premise{}
			\inferrule*[right=\rid]
			\premise
			{\oftp{\Gamma,x:A}{x}{A}}
			\and
			\inferrule*[right=$\rightarrow$-\rintro]
			{\oftp{\Gamma,\tmtp xA}{t}{B}}
			{\oftp\Gamma{\lam{x:A} t}{A\rightarrow B}}
			\\
			\inferrule*[right=$\rightarrow$-\relim]
			{\oftp\Gamma{t}{A\rightarrow B} \\
				\oftp\Gamma{u}{A}}
			{\oftp\Gamma{tu}{B}}
			\and
			\inferrule*[right=$\forall$-$\rintro^*$]
			{\oftp{\Gamma}{t}{A}}
			{\oftp{\Gamma}{\Lambda X.t}{\forall X.A}}
			\and
			\inferrule*[right=$\forall$-\relim]
			{\oftp\Gamma{t}{\forall X.A}}
			{\oftp\Gamma{tB}{A[B/X]}}
		\end{mathparpagebreakable}
		${}^*$ where $X\notin\textrm{FV}(\Gamma)$

\end{document}